\documentclass[a4paper,11pt]{article}

\usepackage{latexsym}
\usepackage[latin1]{inputenc}
\usepackage{amsmath,amssymb,amsfonts,amsthm}
\usepackage{graphicx}
\usepackage{geometry}
\usepackage[english]{babel}
\usepackage{url}
\usepackage{stmaryrd}
\usepackage{fixltx2e}

\usepackage{amsfonts}

\usepackage{array}
\usepackage{hyperref}

\newcommand{\beq}{\begin{equation}}
\newcommand{\eeq}{\end{equation}}
\newcommand{\bi}{\begin{itemize}}
\newcommand{\bd}{\begin{description}}
\newcommand{\ei}{\end{itemize}}
\newcommand{\ed}{\end{description}}

\newcommand{\bc}{\begin{center}}
\newcommand{\ec}{\end{center}}

\newtheorem{df}{Definition}[section]
\newtheorem{ass}[df]{Assumption}
\newtheorem{lm}[df]{Lemma}
\newtheorem{prop}[df]{Proposition}
\newtheorem{thm}[df]{Theorem}
\newtheorem{cor}[df]{Corollary}
\newtheorem{rem}[df]{Remark}
\theoremstyle{definition}
\newtheorem{ex}[df]{Example}

\numberwithin{equation}{section}

\newcommand{\F}{{\mathcal{F}}}

\newcommand{\QG}{{\mathbb{Q}}}
\newcommand{\PG}{{\mathbb{P}}}
\newcommand{\RG}{{\mathbb{R}}}
\newcommand{\PC}{{\mathcal{P}}}
\newcommand{\QC}{{\mathcal{Q}}}
\newcommand{\EC}{{\mathcal{E}}}
\newcommand{\FG}{{\mathbb{F}}}
\newcommand{\FC}{{\mathcal{F}}}
\newcommand{\GC}{{\mathcal{G}}}
\newcommand{\GG}{{\mathbb{G}}}
\newcommand{\HC}{{\mathcal{H}}}

\newcommand{\ind}[1]{{\bf{1}}_{\{{#1}\}}}
\newcommand*\samethanks[1][\value{footnote}]{\footnotemark[#1]}

\DeclareMathOperator*{\esssup}{ess\,sup}
\DeclareMathOperator*{\essinf}{ess\,inf}

\begin{document}

\title{Robust Financial Bubbles}
\author{Francesca Biagini\thanks{Workgroup Financial and Insurance Mathematics, Department of Mathematics, Ludwig-Maximilians Universit\"at, Theresienstra\ss e 39, 80333 Munich, Germany. Emails: biagini@math.lmu.de, mancin@math.lmu.de} $\qquad$        Jacopo Mancin\samethanks
 }

\date{\today}

\maketitle

\begin{abstract} 
We study the concept of financial bubble in a market model endowed with a set $\PC$ of probability measures, typically mutually singular to each other. In this setting we introduce the notions of robust bubble and robust fundamental value in a consistent way with the existing literature in the case $\PC=\{\PG\}$. The notion of no dominance is also investigated under the uncertainty framework. Finally, we provide concrete examples illustrating our results. 
\vspace{0.2cm}

\noindent 
\textbf{Keywords}: Financial bubbles, model uncertainty. \vspace{0.1cm}

\noindent
\textbf{AMS classification}: 60G48, 60G07, 91G99.
\end{abstract}

\section{Introduction}
The mathematical modelization of bubbles has attracted an increasing attention in the recent years. Despite the large and different literature on this topic, the description of a financial bubble is usually built on two main elements: the market value of an asset and its fundamental price. As the first is simply observed by the agent, the modeling of the intrinsic value is where the differences between different approaches arise. In the classical setup, where models with one prior are considered, the main approach is given by the martingale theory of bubbles (we cite for example \cite{sorin}, \cite{protter_complete}, \cite{protter_incomplete} and \cite{protter_fundamental}). According to this theory the fundamental value is defined as the expected sum of future discounted payoffs. Recently another definition has been proposed in \cite{schweizer_bubbles}, where the fundamental value is assumed to be the superreplication price of the asset.\\
The common base of all this models is the starting choice of a filtered probability space and the use of one prior. In this paper we aim to contribute to the existing literature by proposing a framework for the formation of bubbles in a continuous-time financial market under uncertainty. By doing this we allow the investor to consider a wider set of models and to make robust decisions with respect to the scenarios contemplated by all of them. The market value will still be exogenous, but the fundamental value will have a fairly different interpretation from the classical literature and generate unexpected consequences.

We suppose the agent to be endowed with a family $\PC$ of local martingale measures, each one specifying possible dynamics of the financial assets. The set of priors satisfies some regularity conditions, stated in Assumption \ref{ass}, according to the results in \cite{nutz_constructing}. Among the variety of frameworks in model uncertainty literature (see for example \cite{cohen}, \cite{Peng:Gexp}, \cite{soner_aggregation} and \cite{Vorbrink}), this choice is motivated by several reasons. This model guarantees the existence of a time consistent sublinear expectation and has been used in many recent works for the study of some of the classical problems in stochastic finance, such as the absence of arbitrage and superreplication prices (see for example \cite{nutz_robust}, \cite{nutz_jumps}, \cite{nutz_superhedging} and \cite{nutz_optimalstopping}). Moreover it includes two of the most interesting volatility uncertainty models, namely the $G$-setting (see \cite{Peng:Gexp}) and the random $G$-setting (see \cite{nutz_random}), whilst allowing for the tractability of stopping times.

The concept of bubble is more delicate to define in the presence of uncertainty. One of the problems consists in providing a well-posed definition of \emph{robust fundamental value} $S^\ast=(S^\ast_t)_{t\geq 0}$ of a given financial asset $S=(S_t)_{t\geq 0}$: the \emph{robust bubble} $\beta=(\beta_t)_{t\geq 0}$ will be again the difference between $S$ and $S^\ast$. Not only the fundamental value needs to be consistent with the existing literature when $\PC=\{\PG\}$, but it must also rule out trivial situations, e.g.\ a bubble in an underlying $\PG$-market determining the presence of robust bubble. Since in this setting we have no linear pricing system and because of the consequent difficulty to transpose in the present model the concept of risk neutral valuation of discounted future payments, we choose to describe robust fundamental values through superreplication prices, see Definition \ref{rfv}, by extending the approach of \cite{schweizer_bubbles}. In Section \ref{RFV} we accurately discuss this issue.\\
One of the main novelties of our approach is then the $\PG$-\emph{local submartingale} behavior of bubbles under each $\PG$-market, i.e.\ when examined under the view of an agent endowed with just one prior $\PG$ from the family $\PC$. This generalizes in a natural way the local-martingale dynamics displayed in the classical models from the literature and it allows to describe the birth of a bubble and its growth in size in a \emph{static model}, i.e.\ without changing the investor's views on the market over time. The same submartingale behavior is in fact described for some cases also in \cite{sorin}, but it is the result of a smooth shift from a pricing measure to another. To the best of our knowledge this description of bubbles is new, as it distinguishes itself also from the robust setting outlined in \cite{cox_robust}, where bubbles arise as a consequence of constraints on possible trading strategies in a different setup.\\
Another interesting feature of our model is the way a robust bubble is perceived in the underlying $\PG$-markets: it might in fact happen that a bubble is not seen as such for some particular priors. Alternatively stated, the asset originating the robust bubble may be a true $\PG$-martingale for some $\PG\in\PC$. To this regard our results represent a relevant extension of the setting of \cite{schweizer_bubbles}, where market bubbliness excludes the existence of a true martingale measure.\\
Finally we investigate the concept of no dominance, proposing its robust counterpart in the model uncertainty framework and studying its consequences on the concept of robust bubble.

The paper is organized as follows. In Section 2 we outline the notation and the financial model. In particular we will recall the results from \cite{nutz_constructing} and construct a new set of priors satisfying Assumption \ref{ass}. In Section 3, after reviewing the existing literature, we discuss and study our concept of robust bubble and \emph{robust no dominance}, illustrating our results through concrete examples. In Section 4 we conclude by examining the situation in which the time horizon is not bounded.

\section{The Setting}\label{setting}
We consider a financial market under a family $\PC$ of probability measures, typically non-dominated, on $\Omega=C_0(\RG_+,\RG^d)$, the space of continuous paths $\omega=(\omega_s)_{s\geq0}$ in $\RG^d$ with $\omega_0=0$ endowed with the topology of locally uniform convergence. We denote with $\FC$ the Borel $\sigma$-field on $\Omega$. We are interested in sublinear expectations 
\[
\xi\mapsto \EC_0(\xi):=\sup_{\PG\in\PC} E_\PG[\xi],
\]
inducing time consistent conditional sublinear expectations. For this reason some conditions have to be enforced both on the set of priors and on the random variables we take into account. Given a stopping time $\tau$ of the filtration $\FG:=\{\F_t\}_{t\geq0}$ generated by the canonical process, the main technical issue is to guarantee that 

\beq
\EC_\tau(\xi)=\esssup_{\PG^\prime \in \PC(\tau,\PG)} E_{\PG^\prime} [\xi|\F_\tau]\qquad \PG-a.s.\text{ for all } \PG\in\PC,
\eeq
where $\PC(\tau,\PG)=\{\PG^{\prime}\in\PG\; : \; \PG^{\prime}=\PG\text{ on } \F_\tau\}$, is well-defined as conditional sublinear expectation operator. This problem is solved in the literature by means of different approaches, generally by shrinking the set of priors $\PC$ or by requiring strong regularity of the random variables. We cite \cite{cohen}, \cite{nutz_superhedging}, \cite{nutz_constructing}, \cite{peng_nonlinear} and \cite{soner_aggregation} to mention some of the papers on this topic.
We choose to place ourselves in the context of \cite{nutz_constructing} as it generalizes the frameworks of $G$-expectation and random $G$-expectation and provides some tractability of stopping times, which remains still an open question in the $G$-setting. \\
For the sake of completeness we then summarize the hypothesis we enforce on the set $\PC$, as stated in \cite{nutz_constructing}, together with the notation introduced thereby. Let $\mathfrak{P}(\Omega)$ be the set of all probability measures on $(\Omega,\FC)$ equipped with the topology of weak convergence. For any stopping time $\tau$, the concatenation of $\omega,\tilde{\omega}\in\Omega$ at $\tau$ is the path
\[
(\omega\otimes_\tau\tilde{\omega})_u:=\omega_u\textbf{1}_{[0,\tau(\omega))}(u)+(\omega_{\tau(\omega)}+\tilde{\omega}_{u-\tau(\omega)})\textbf{1}_{[\tau(\omega),\infty)}(u),\quad u\geq0.
\]
Given a function $\xi$ on $\Omega$ and $\omega\in\Omega$, we define the function $\xi^{\tau,\omega}$ on $\Omega$ by
\[
\xi^{\tau,\omega}(\tilde{\omega}):=\xi(\omega\otimes_\tau\tilde{\omega}),\quad \tilde{\omega}\in\Omega.
\]

\noindent
For any probability measure $\PG\in\mathfrak{P}(\Omega)$ there exists a regular conditional probability distribution $\{\PG_\tau^\omega\}_{\omega\in\Omega}$ given $\FC_\tau$. That is $\PG_\tau^\omega\in\mathfrak{P}(\Omega)$ for each $\omega$, while $\omega\mapsto \PG_\tau^\omega(A)$ is $\FC_\tau$-measurable for any $A\in\FC$ and
\[
E_{\PG_\tau^\omega}[\xi]=E_{\PG}[\xi|\FC_\tau](\omega)\quad\text{for}\quad \PG-a.e.\;\omega\in\Omega,
\]
whenever $\xi$ is $\FC$-measurable and bounded. Moreover, $\PG_\tau^\omega$ can be chosen to be concentrated on the set of paths that coincide with $\omega$ up to time $\tau(\omega)$,
\[
\PG_\tau^\omega\{\omega^\prime\in\Omega:\; \omega^\prime=\omega\text{ on } [0,\tau(\omega)]\}=1\quad \text{for all } \omega\in\Omega.
\]
We define the probability measure $\PG^{\tau,\omega}\in\mathfrak{P}(\Omega)$ by
\[
\PG^{\tau,\omega}(A):=\PG_\tau^\omega(\omega\otimes_\tau A),\quad A\in\FC,\quad \text{where } \omega\otimes_\tau A:=\{\omega\otimes_\tau\tilde{\omega}:\; \tilde{\omega}\in A\}.
\]
We then have the the identities
\[
E_{\PG^{\tau,\omega}}[\xi^{\tau,\omega}]=E_{\PG_\tau^\omega}[\xi]=E_\PG[\xi|\FC_\tau](\omega)\quad \text{for}\quad \PG-a.e.\;\omega\in\Omega.
\]

\noindent
We next recall, as in \cite{nutz_constructing}, some basics from the theory of analytic sets. A subset of a Polish space is called \emph{analytic} if it is the image of a Borel subset of another Polish space under a Borel-measurable  mapping. In particular any Borel set is analytic. The collection of analytic sets is stable under countable intersections and unions, but in general not under complementation. Moreover for every $t\geq 0$ the universal completion of $\FC_t$ is the $\sigma$-field $\FC_t^\ast=\cap_\PG \FC_t^\PG$, where $\PG$ ranges over all probability measures on $\FC_t$ and $\FC_t^\PG$ is the completion of $\FC_t$ under $\PG$. We denote with $\FG^\ast$ the filtration $\{\FC^\ast_t\}_{t\geq 0}$.

\noindent
For each $(s,\omega)\in\RG_+\times\Omega$ we fix a set $\PC(s,\omega)\subseteq \mathfrak{P}(\Omega)$. Assume that 
\[
\PC(s,\omega)=\PC(s,\tilde{\omega})\qquad \text{if}\qquad \omega|_{[0,s]}=\tilde{\omega}|_{[0,s]}.
\]
We then state Assumption 2.1 from \cite{nutz_constructing}.

\begin{ass}
\label{ass}
Let $(s,\bar{\omega})\in\RG_+\times\Omega$, let $\tau$ be a stopping time such that $\tau\geq s$ and $\PG\in\PC(s,\bar{\omega})$. Set $\theta:=\tau^{s,\bar{\omega}}-s$.
\begin{itemize}
\item[(i)] \emph{Measurability}: The graph $\{(\PG^\prime,\omega):\; \omega\in\Omega, \; \PG^\prime\in\PC(\tau,\omega)\}\subseteq \mathfrak{P}(\Omega)\times \Omega$ is analytic.
\item[(ii)] \emph{Invariance}: We have $\PG^{\theta,\omega}\in\PC(\tau, \bar{\omega}\otimes_s \omega)$ for $\PG$-a.e\ $\omega\in\Omega$.
\item[(iii)] \emph{Stability under pasting}: If $\nu:\Omega\to\mathfrak{P}(\Omega)$ is a $\F_\theta$-measurable kernel and $\nu(\omega)\in\PC(\tau, \bar{\omega}\otimes_s \omega)$ for $\PG$-a.e\ $\omega\in\Omega$, then the measure defined by 
\beq
\label{cond21}
\bar{\PG}(A)=\int\int ({\bf{1}}_{A})^{\theta,\omega}(\omega^\prime)\nu(d\omega^\prime;\omega)\PG(d\omega),\qquad A\in\F
\eeq
is an element of $\PC(s,\bar{\omega})$.
\end{itemize}
\end{ass}

\noindent 
Exploiting the previous conditions, Theorem 2.3 in \cite{nutz_constructing} proves the following.
\begin{thm}\label{robustsetting}
Let $\sigma\leq\tau$ be stopping times and $\xi:\Omega\to\bar{\RG}$ be an upper semianalytic function. Then under Assumption \ref{ass} the function
\[
\EC_\tau(\xi)(\omega):=\sup_{\PG\in \PC(\tau,\omega)} E_\PG[\xi^{\tau,\omega}],\quad \omega\in\Omega
\]
is $\FC_\tau^\ast$-measurable and upper semianalytic. Moreover
\beq
\label{nutz1}
\EC_\sigma(\xi)(\omega)=\EC_\sigma(\EC_\tau(\xi))(\omega)\quad \text{for all}\quad \omega\in\Omega.
\eeq
Furthermore,
\beq
\label{nuzt2}
\EC_\tau(\xi)=\esssup_{\PG^\prime \in \PC(\tau,\PG)} E_{\PG^\prime}[\xi|\FC_\tau]\quad \PG-a.s.\quad \text{for all}\quad \PG\in\PC,
\eeq
where $\PC(\tau,\PG)=\{\PG^\prime\in\PC:\; \PG^\prime=\PG\text{ on } \FC_\tau\}$, and in particular
\beq
\label{nutz3}
\EC_\sigma(\xi)=\esssup_{\PG^\prime\in\PC(\sigma,\PG)} E_{\PG^\prime}[\EC_\tau(\xi)|\FC_\sigma]\quad \PG-a.s.\quad \text{for all}\quad \PG\in\PC.
\eeq
\end{thm}

\noindent
We finally call $\PC$-martingale, an adapted stochastic process $M=(M_s)_{s\geq 0}$ such that $\EC_0(M_t)$ is finite for every $t$ and 
\[
M_t = \EC_t(M_T)
\]
for any $T\geq t$. The particular $\PC$-martingales for which also $-M$ is a $\PC$-martingale are called $\PC
$-symmetric martingales.

It is an important result of \cite{nutz_constructing} that the $G$-expectation framework can be incorporated in the model described above. More precisely, consider the set of martingale measures 
\[
\mathfrak{M}=\{\PG\in\mathfrak{P}(\Omega):\; B \text{ is a local $\PG$-martingale}\},
\]
where $B=\{B_u(\omega)\}$ denotes the canonical process, and its subset
\[
\mathfrak{M}_a=\{\PG\in\mathfrak{M}:\; \langle B\rangle^\PG \text{ is absolutely continuous $\PG$-a.s.}\},
\] 
where now $\langle B\rangle^\PG$ is the $\RG^{d\times d}$-valued quadratic variation process of $B$ under $\PG$  and absolute continuity refers to the Lebesgue measure. We report here Proposition 3.1 from \cite{nutz_constructing}.
\begin{prop}\label{G}
The set 
\[
\PC_{\textbf{D}}=\{\PG\in\mathfrak{M}_a:\; d\langle B\rangle_t^\PG/dt\in \textbf{D}\;\; \PG\times dt-a.e.\},
\]
where $\textbf{D}$ is a nonempty, convex and compact subset of $\RG^{d\times d}$, satisfies Assumption \ref{ass}.
\end{prop}

\noindent
It is indeed well known that the sublinear expectation 
\[
\EC_0^{\textbf{D}}(\xi):=\sup_{\PG\in\PC_{\textbf{D}}}E_\PG[\xi]
\]
yields the $G$-expectation on the space of quasi continuous functions in $\mathbb{L}^1_G$. We prove in the next proposition that the set $\PC_{\textbf{D},const}\subset\PC_\textbf{D}$ of \emph{constant volatility scenarios} satisfies Assumption \ref{ass}. This result plays a key role for the examples of Section 3.

\begin{prop} \label{const}
The set 
\[
\PC_{\textbf{D}, const}=\{\PG\in\PC_\textbf{D}:\; d\langle B\rangle_t^\PG/dt \text{ is constant for all } t\;	\; \PG-a.s.\},
\]
where $\textbf{D}$ is a nonempty, convex and compact subset of $\RG^{d\times d}$, satisfies Assumption \ref{ass}.
\end{prop}

\begin{proof}
We divide the proof in three steps, following those in Theorem 4.3.\ in \cite{nutz_constructing} and using the same notation. \\
\emph{Step 1}: Lemma 4.4.\ in \cite{nutz_constructing} shows that $\mathfrak{M}_a$ is Borel-measurable by proving that 
\[
\mathfrak{M}_a=\left\{\PG\in\mathfrak{M}:\;\langle B\rangle_t=\int_0^t\varphi_s ds\;\; \PG-a.s.\text{ for all } t\in\mathbb{Q}_+\right\},
\]
where $\varphi$ is Borel-measurable and corresponds $\PG$-a.s.\ to the density of the absolutely continuous part of $\langle B\rangle$ with respect to the Lebesgue measure. Analogously it can be shown that 
\[
\PC_{\textbf{D}, const}=\left\{\PG\in\mathfrak{M}:\;\langle B\rangle_t=\int_0^t\varphi_s ds=ct\;\; \PG-a.s.\text{ for all } t\in\mathbb{Q}_+,\; c\in\textbf{D}\right\}.
\]
Hence $\PC_{\textbf{D}, const}$ is Borel-measurable, which is enough to guarantee the measurability required in Assumption \ref{ass}. \\
\emph{Step 2}: Let now $\tau$ be a stopping time and $\PG\in\mathfrak{M}_a$. For $\PG$-a.e.\ $\omega\in\Omega$, we have $\PG^{\tau,\omega}\in\mathfrak{M}_a$, as shown in Lemma 4.7.\ in \cite{nutz_constructing}. But then, if $\QG\in\PC_{\textbf{D},const}$, we have $\QG^{\tau,\omega}\in \PC_{\textbf{D}, const}$, as $d\langle B\rangle_t^\QG/dt$ is constant for all $t$ $\QG^{\tau,\omega}$-a.s. being a regular conditional probability distribution of $\QG$, which ensures condition (ii) of Assumption \ref{ass}. \\
\emph{Step 3}: To prove the validity of condition (iii) we need to introduce some more notation. Let $s\in\RG_+$, $\tau\geq s$ be a stopping time and $\bar{\omega}\in\Omega$ and $\PG\in\PC_{\textbf{D},const}(s,\bar{\omega})$. Moreover, let $\theta:=\tau^{s,\bar{\omega}}-s$, let $\nu:\;\Omega\mapsto\mathfrak{P}(\Omega)$ be an $\FC_\theta$-measurable kernel such that $\nu(\omega)\in\PC_{\textbf{D},const}(\tau,\bar{\omega}\otimes_s\omega)$ for $\PG$-a.s. $\omega\in\Omega$ and let $\bar{\PG}$ be defined as in \eqref{cond21}. Finally, denote with $\hat{a}$ the $\bar{\RG}^{d\times d}$-valued process 
\beq
\label{nutza}
\hat{a}_t(\omega):=\limsup_{n\to\infty} n \left[ \langle B\rangle_t(\omega)-\langle B\rangle_{t-\/n}(\omega)\right],\quad t>0.
\eeq

\noindent
We need to show that $\bar{\PG}\in\PC_{\textbf{D},const}(s,\bar{\omega})$. To this end, we notice that because of Lemma 4.9.\ in \cite{nutz_constructing}, the only thing we have to prove is 
\[
\left(dr\times \nu(\omega)\right)\left\{ (r,\omega^\prime)\in \llbracket \theta(\omega),\infty\llbracket:\; \hat{a}_r(\omega^\prime) \notin \textbf{D},\;\hat{a}_\cdot(\omega^\prime) \text{ is not constant }  \right\}=0,
\]
for $\PG$-a.e.\ $\omega\in\Omega$, but this is guaranteed by the fact that 
\[
\left(dr\times \nu(\omega)\right)\left\{ (r,\omega^\prime)\in \llbracket \theta(\omega),\infty\llbracket:\; \hat{a}_r(\omega^\prime) \notin \textbf{D}\right\}=0,
\]
as shown in \cite{nutz_constructing}.
\end{proof}

\subsection{The Market Model}\label{marketmodel}
We assume that our market model is given by $(\Omega,\FC)$ endowed with a set of LMM $\QC$ satisfying Assumption \ref{ass}, as introduced above. We consider a discounted risky asset given by a $\RG^d$-valued, $\FG^\ast$-adapted and right-continuous process $S=(S_t)_{t\geq 0}$ such that its paths are $\QC$-q.s.\ continuous. Let $\tau>0$ q.s.\ be a stopping time describing the maturity of the risky asset and $X_\tau$ be the final payoff or liquidation value at time $\tau$. The bank account $S^0$ is assumed to be constant and equal to $1$. The wealth process $W=(W_t)_{t\geq 0}$ generated from owning the asset is given by 
\[
W_t:= S_t\ind{\tau>t} +X_\tau\ind{\tau\leq t}.
\]

\noindent
In the standard literature on bubbles it is usually assumed that the No Free Launch With Vanishing Risk condition (NFLVR) holds. When working in the context of multiple priors models the situation becomes more involved. In fact there does not exist a robust counterpart to NFLVR yet. There is actually just one well studied concept of arbitrage (arbitrage of the first kind (NA\textsubscript{1})) in the continuous time setting under uncertainty introduced in \cite{nutz_robust}. However in \cite{nutz_robust}, the existence of absolutely continuous martingale measures requires to introduce a stopping time $\zeta$ that causes a jump to a cemetery state, which is invisible under all $\PG\in\PC$ but may be finite under some $\QG\in\QC$, where $\QC$ is an appropriate set of local martingale measures. Therefore, despite the possibility to start with a family $\PC$ of physical measures, the results of \cite{nutz_robust} require to reserve some particular care to the tractability of $\zeta$. This is one of the reasons why, while working in the setting outlined, we will assume for simplicity the following.

\begin{ass}\label{NFLVR}
We assume that the wealth process is a $\QG$-local martingale for every $\QG\in\QC$. Thus the set $\QC$ is made of LMM, enforcing NFLVR under all $\QG$-market. 
\end{ass}

\noindent
By doing this we guarantee at the same time that $W$ is economically justified under all probability scenarios.\\
Moreover, as in the classical setting NFLVR implies NA\textsubscript{1}, and the fact robust NA\textsubscript{1} implies NA\textsubscript{1} under all priors included in the uncertainty framework (see \cite{nutz_robust}), it is reasonable to expect that a robust version of NFLVR will also imply the correspondent robust NA\textsubscript{1}. This question, which is interesting and complex on its own, is beyond the aim of this paper.

We next describe which are the trading strategies allowed in the market. Denote, as in \cite{nutz_jumps}, $L(S,\QC)$ the set of all $\mathbb{R}^d$-valued, $\FG$-predictable processes that are $S$-integrable for all $\QG\in\QC$. Let further $\mathcal{N}^\QC$ be the collection of sets that are $(\FC,\QG)$-null for all $\QG\in\QC$. Denote $\GG=(\GC_t)_{t\geq0}$, where
\[
\GC_t:=\F_t^\ast\vee \mathcal{N}^\QC.
\]

\begin{df}
We then say that a $\GG$-predictable process $H\in L(S,\QC)$ is \emph{admissible} if $H\cdot S$ is a $\QG$-supermartingale for every $\QG\in\QC$, and we denote $\mathcal{H}$ the sets of all such processes.
\end{df}

\section{Robust Bubbles}
An important part of the literature defines a bubble as the situation in which the \emph{market value} $S$ of an asset is greater than its \emph{fundamental value} $S^\ast$. In other words a bubble appears at time $t$ if 
\[
\beta_t:=S_t-S^\ast_t>0.
\]

\noindent
In order to have a better understanding of what should be the right notion of asset fundamental value under model uncertainty, we first present a short survey on how this concept is modeled in the classical literature of financial bubbles. \\
For simplicity we will start by considering a finite time horizon. Let then $T\in\mathbb{R}_+$ be such that $\tau\leq T$. We note that in this case $W_t=S_t$ for every $t\in[0,T]$, if $X_\tau=S_\tau$, which will be assumed throughout this section.

\subsection{Classical Fundamental Value Modeling}\label{classicsec}
When a unique prior $\PG$ exists, we could recognize two main approaches for defining the fundamental value of a financial asset. In one setting, see \cite{sorin}, \cite{protter_complete} and \cite{protter_incomplete} for a reference, this is defined as the asset's discounted future payoffs under a risk neutral measure. This means that if $\mathbb{Q}\in \mathcal{M}_{loc}(S)$, where $\mathcal{M}_{loc}(S)$ denotes the set of all equivalent martingale measures for $S$, the fundamental value $S^\ast=(S^\ast_t)_{t\in[0,T]}$ is given by
\[
S^\ast_t= E_\QG [S_T | \F_t]
\]
for every $t\in[0,T]$, where $T$ is a fixed finite time horizon. The concept of bubble depends on the following distinction
\[
\mathcal{M}_{loc}(S)=\mathcal{M}_{UI}(S)\cup\mathcal{M}_{NUI}(S),
\]
where $\mathcal{M}_{UI}(S)$ is the class of measures $\QG\approx\PG$ such that $S$ is a uniformly integrable martingale under $\QG$ and $\mathcal{M}_{NUI}(S)=\mathcal{M}_{loc}(S)\setminus \mathcal{M}_{UI}(S)$.\\
The market bubbliness thus is built upon the investor's views: if she acts accordingly to a $\mathbb{Q}\in\mathcal{M}_{UI}(W)$ then she would see no bubble; on the contrary, if $\mathbb{R}\in\mathcal{M}_{NUI}(W)$ is perceived to be the right market view then there would be an asset bubble. In this sense the concept of bubble is dynamic: bubbles are born or burst depending on how the investor changes her perspectives on the market.\\
If the market is complete the situation simplifies. As $\mathcal{M}_{loc}(S)$ is made of a unique element (and it must exist if the usual NFLVR condition holds), either there is a bubble from the beginning or there is no bubble at all. This result agrees with the second main approach to financial bubbles (see \cite{schweizer_bubbles} and the references therein), in which the fundamental value coincide with the superreplication price. Denoting by  given by $\mathbb{L}^0_+(\FC_t)$ the set of $\FC_t$-measurable random variables taking $\PG$-a.s.\ values in $[0,\infty)$, the superreplication price is given by
\[
\pi_t(S):=\essinf\{v\in\mathbb{L}^0_+(\FC_t)\;:\;\exists \;\theta\in\Theta_t \text{ with } v+(\theta\cdot S)_T\geq S_T\;\PG-a.s. \}
\] 
for $t\in[0,T]$ and a suitably defined class of investment strategies $\Theta_t$ that are different from $0$ only on the interval $[t,T]$. Given the duality (see \cite{kramov})
\beq
\label{superreplication}
\pi_t(S)=\sup_{\QG\in\mathcal{M}_{loc}(S)}E_\QG[S_T|\F_t],
\eeq
if there exists a bubble in a complete market then the superreplication price must be lower than the actual price of the asset observed in the market.\\
Differences between the two approaches emerge in the context of incomplete markets. Given the superreplication duality \eqref{superreplication}, as soon as $\mathcal{M}_{UI}(S)\neq \emptyset$ then there is no bubble. The concepts of bubble itself and bubble birth change. It might be that the superreplication price and the market price are equal at $t=0$ but they may differ at a later time $t>0$ (see Example 3.7 in \cite{schweizer_bubbles}): at the time $t$ is the bubble is born. Here the `bubble missprice' on an asset is given by the fact that we can get the same wealth at the end but exploiting an investment strategy with lower initial price. Still we cannot profit from it given the set of admissible strategies: we need to go short on the asset and long on the superreplicating strategy to generate a sure profit at terminal time but by doing this we also face the risk of unbounded losses in $(0,T)$. In this setting, if a bubble exists then it is perceived by any investor, independently from the particular $\QG\in\mathcal{M}_{loc}(S)$ chosen to price contingent claims.

Finally we present an interesting result that links the two settings described above. We prove that if there is no $\QG\in\mathcal{M}_{loc}(S)$ that excludes the presence of a bubble in the sense of \cite{sorin} or \cite{protter_incomplete}, then there is also a bubble in the case fundamental values are given by superreplication prices. 

\begin{prop}\label{protter vs schweizer}
Let $S=(S_t)_{t\in[0,T]}$ be a continuous adapted process in a filtered probability space satisfying the usual conditions. If $\mathcal{M}_{loc}(S)=\mathcal{M}_{NUI}(S)$ then there is a $t\in[0,T)$ such that 
\[
S_t>\sup_{\QG\in\mathcal{M}_{loc}(S)}E_\QG[S_T|\F_t].
\]
\end{prop}

\begin{proof}
We argue by contradiction. If we suppose that
\[
S_t=\sup_{\QG\in\mathcal{M}_{loc}(S)}E_\QG[S_T|\F_t]
\]
for every $t\in [0,T]$, the process $\pi(S)$ defined in \eqref{superreplication} is a $\QG$-local martingale for each $\QG\in\mathcal{M}_{loc}(S)$. This implies, according to Theorem 3.1 in \cite{kramov}, that the minimal superreplicating portfolio is \emph{self-financing} and thus that there exists $\bar{\QG}\in\mathcal{M}_{loc}(S)\cap \mathcal{M}_{UI}(S)$, contradicting the hypothesis.
\end{proof}

\begin{rem}
It is possible to obtain the same result also in the \emph{dominated case}, i.e.\ when in place of considering $\mathcal{M}_{loc}(S)$ we look at
\[
\QC=\{\QG\in\mathfrak{P}(\Omega)\;|\; \QG\ll \PG,\; S\text{ is a $\QG$-local martingale}\}.
\]
In this context, if $\QC$ is \emph{m-stable}, it is possible to define the asset fundamental value as
\[
S_t^\ast= \esssup_{\QG\in\QC} E_\QG[S_T|\FC_t],
\]
for any $t\in[0,T]$ (for this result and the definition of m-stability we refer to \cite{delbaen}). However, as the measures in $\QC$ which are equivalent to $\PG$ are dense in $\QC$ (see again \cite{delbaen}), it holds
\[
S_t^\ast= \esssup_{\QG\in\QC} E_\QG[S_T|\FC_t]=\sup_{\QG\in\mathcal{M}_{loc}(S)}E_\QG[S_T|\F_t],
\]
so that Proposition \ref{protter vs schweizer} also applies also in this case.
\end{rem}

\subsection{Robust Fundamental Value}\label{RFV}
We start this section with some considerations, fixing which are the requirements we demand to our robust model for bubbles. They are the following:
\begin{itemize}
\item[(i)] When $\mathcal{Q}$ boils down to a singleton the model must collapse into one of the two approaches mentioned in the previous section. This already tells us that the robust fundamental value should be defined in terms of some conditional expectation.
\item[(ii)] If a bubble is detected under one $\mathbb{Q}\in\mathcal{Q}$, this does not have to automatically mean that the bubble is robust.
\end{itemize}

\noindent
The initial step required is the introduction of a robust concept for the fundamental value that can be meaningfully used in our setting. A first try, in analogy to what happens to arbitrage, is to define asset bubble the case when there exists a $\mathbb{Q}\in\mathcal{Q}$ such that the fundamental value under $\mathbb{Q}$ is lower than the market value, whilst being (lower or) equal under all other priors. This naive first definition is not well-posed as any classical bubble will be turn into a robust bubble. To overcome this problem we could decide to define a `$\mathbb{Q}$-fundamental value' under each prior. To be consistent with the existing literature and recover the traditional setup of \cite{protter_incomplete} when $\mathcal{Q}$ consists of only one probability measure, we could define 
\[
S^{\ast,\mathbb{Q}}_t=E_{\mathbb{Q}}[S_T|\mathcal{F}_t]
\]
to be the fundamental value under $\mathbb{Q}\in\mathcal{Q}$. This class of $\mathbb{Q}$-fundamental values will not eventually be aggregable (this is already the case with the $G$-setting, see \cite{soner_representation}). Intuition suggests to define a bubble the situation in which 
\[
\mathbb{Q}(S^{\ast,\mathbb{Q}}_t<S_t)>0
\]
for each $\mathbb{Q}\in\mathcal{Q}$ and some $t>0$. Alternatively stated we would say that the asset $S$ is a $\mathcal{Q}$-bubble if it is a $\mathbb{Q}$-bubble for every $\mathbb{Q}\in\mathcal{Q}$. It is natural to run here a parallel with the notion of arbitrage. In \cite{Vorbrink} a \emph{robust arbitrage} is defined as a trading strategy that requires zero initial wealth, but excludes losses quasi surely (i.e. $\mathbb{Q}$-a.s.\ for all $\mathbb{Q}\in\mathcal{Q}$) and delivers a positive gain with positive probability for at least one $\mathbb{Q}\in\mathcal{Q}$. It would be possible to strengthen this definition by further requiring that a robust arbitrage should generate profit with positive probability for all $\mathbb{Q}\in\mathcal{Q}$, but this would be too strong. This is precisely the same problem hidden in our first definition of robust bubble. \\
There is also a deeper issue, which is peculiar to the nature of our framework. In a market model that is intrinsically incomplete, it is not immediate to transpose the concept of `expected future payoffs' from the classical setting to the modeling under uncertainty, because of the absence of a linear pricing system. The first naive definition of robust fundamental value that we propose above is actually linked to this notion, as it coincides with the approach of \cite{protter_incomplete} when $\QG$ reduces to a singleton. \\
This definition of fundamental value as superreplication price provides a clear financial interpretation and is more suitable for this setting. 
\begin{df}
\label{rfv}
We call robust fundamental value the process $S^\ast=(S_t^\ast)_{t\in[0,T]}$ where
\beq
\label{rfve}
S^{\ast}_t=\esssup_{\mathbb{Q}^\prime\in \mathcal{Q}(t,\mathbb{Q})} E_{\mathbb{Q}^\prime}[S_T|\mathcal{F}_t],\quad \QG-a.s.
\eeq
for every $\mathbb{Q}\in\mathcal{Q}$, where $\mathcal{Q}(t,\mathbb{Q})=\{\QG^\prime\in\mathcal{Q}:\;\mathbb{Q}^\prime=\mathbb{Q} \text{ on } \mathcal{F}_t\}$. There is a \emph{robust bubble} if there exists a stopping time $\tau$ such that 
\[
\QG(S_\tau > S_\tau^{\ast})>0
\]
for a $\mathbb{Q}\in\mathcal{Q}$. We denote the robust bubble by $\beta=(\beta_t)_{t\in[0,T]}$, where
\beq
\beta_t := S_t - S_t^{\ast}.
\eeq
\end{df}

\noindent
As opposed to the previous definition, it is not necessary to have a bubble under all scenarios to have a robust bubble. The parallel with the notion of robust arbitrage becomes now evident. As $S$ is a positive $\QG$-local martingale, hence a $\QG$-supermartingale, we have 
\[
S_t\geq S_t^{\ast},\qquad \QG-a.s.
\]
for every $t\in[0,T]$ and $\QG\in\QC$. There is now a robust bubble in the market at a stopping time $\tau$ if there exists a scenario (a probability measure $\bar{\QG}\in\mathcal{Q}$) such that the robust fundamental value is smaller than the market value with positive probability and all probabilities that coincide with $\bar{\QG}$ on $\FC_\tau$ agree on this view. Requirement (ii) in particular becomes more evident: when
\[
\QG(S^{\ast}_t< S_t)>0,
\]
then the same holds for any $\mathbb{Q}^\prime\in \mathcal{Q}(t,\mathbb{Q})$. 

Our definition of robust bubble extends the approach where fundamental prices are given by superreplication prices to the framework under uncertainty. Some extra attention has to be taken on this point. Because of the results in \cite{nutz_jumps}, $S^\ast_0$ can be interpreted in terms of superreplication price if the family $\QC$ is \emph{saturated}, i.e.\ if for every $\QG\in\QC$ all sigma martingale measures equivalent to $\QG$ are contained in $\QC$. In this case
\[
S_0^\ast=\inf\{x\in\RG\;:\;\exists \;H\in\mathcal{H} \text{ with } x+(H\cdot S)_T\geq S_T\; \QG-a.s. \text{ for all } \QG\in\QC\}
\]
and the link with the literature with a unique prior is evident as we show it in the following proposition.

\begin{prop}\label{classicrobust}
Let $\PC=\{\PG\}$, then the robust fundamental value \eqref{rfve} coincides with the classical superreplication price, i.e.
\beq
\label{eq1}
\esssup_{\QG\in\QC} E_\QG[S_T|\F_t]= \esssup_{\QG^\prime\in\QC(t,\QG)} E_{\QG^\prime}[S_T|\F_t]\quad a.s.
\eeq
\end{prop}

\begin{proof}
Note that if $\PC=\{\PG\}$ then $\QC=\{\QG\;|\;\QG\approx \PG,\;\QG \text{ ELMM}\}$ is made of measures equivalent to each other. We prove that  
\beq
\label{eq2}
\esssup_{\QG\in\QC(t,\QG^1)} E_\QG[S_T|\F_t]= \esssup_{\QG^\prime\in\QC(t,\QG^2)} E_{\QG^\prime}[S_T|\F_t]\quad a.s.
\eeq
for every $\QG^1, \QG^2\in\QC$ by a measure pasting argument similar to Proposition 9.1 in \cite{delbaen}. This suffices to conclude as
\[
\esssup_{\QG\in\QC} E_\QG[S_T|\F_t] \geq \esssup_{\QG^\prime\in\QC(t,\QG)} E_{\QG^\prime}[S_T|\F_t],
\]
but \eqref{eq2} also guarantees
\[
\esssup_{\QG\in\QC} E_\QG[S_T|\F_t] \leq \esssup_{\QG\in\QC}\left\{ \esssup_{\QG^\prime\in\QC(t,\QG)} E_{\QG^\prime}[S_T|\F_t]\right\} = \esssup_{\QG^\prime\in\QC(t,\QG)} E_{\QG^\prime}[S_T|\F_t].
\]

\noindent
Assume $\QG^2\in \QC\setminus \QC(t,\QG^1)$, otherwise the claim is trivial. We notice that 
\[
\frac{d\QG}{d\PG}|_{\F_t}=\frac{d\QG^i}{d\PG}|_{\F_t}:=Z_t^i,
\]
for every $\QG\in\QC(t,\QG^i)$, $i=1,2$. This is clear as, for every $A\in\F_t$, it must hold
\[
E_{\PG}[Z_t^i{\bf{1}}_{A}]=E_{\QG^i}[{\bf{1}}_{A}]=\QG^i(A)=\QG(A)=E_{\QG}[{\bf{1}}_{A}]=E_{\PG}\left[\frac{d\QG}{d\PG}|_{\F_t}{\bf{1}}_{A}\right].
\]
Define now
\[
Z_s:= Z_s^1 \ind{s\leq t} + Z_t^1\frac{Z_s^2}{Z_t^2}\ind{t<s},
\]
which is the Radon-Nykodim derivative of an ELMM, as proven in Proposition 9.1 in \cite{delbaen}. The measure $\QG^{\prime}$ associated to $(Z_s)_{s\in[0,T]}$ thus belongs to $\QC(t,\QG^1)$ and satisfies
\begin{align}
\notag E_{\QG^{\prime}}[S_T|\F_t]&=\frac{E_{\PG}[S_TZ_T|\F_t]}{Z_t}=\frac{E_{\PG}\left[S_TZ_t^1\frac{Z_T^2}{Z_t^2}|\F_t\right]}{Z_t^1}\\
\notag &=\frac{E_{\PG}[S_TZ_T^2|\F_t]}{Z_t^2}=E_{\QG^{2}}[S_T|\F_t].
\end{align}
This shows that for every $\QG\in\QC(t,\QG^2)$ there exists a $\QG^\prime\in\QC(t,\QG^1)$ such that $E_{\QG^{\prime}}[S_T|\F_t]= E_{\QG}[S_T|\F_t]$. This is enough to establish \eqref{eq2}. 
\end{proof}

\begin{rem}
Note that in Proposition \ref{classicrobust} we can consider almost sure equalities as $\PC=\{\PG\}$ and we are dealing with the set of martingale measures equivalent to $\PG$.
\end{rem}

\noindent
Saturation is a condition that we do not enforce on our model, but is automatically satisfied if every $\QG$-market is complete. In all other cases the arise of a bubble may be caused either by a difference between the market price and the superreplication price or by a duality gap in 
\[
\sup_{\QG\in\QC} E_\QG[S_T]\leq \inf\{x\in\RG\;:\;\exists \;H\in\mathcal{H} \text{ with } x+(H\cdot S)_T\geq S_T\; \QG-a.s. \text{ for all } \QG\in\QC\}.
\]
This second situation is precisely the one considered in \cite{cox_robust} in order to detect a robust bubble. This means that $S^\ast$ can always be viewed at least as the worst model price, among the models considered by the investor.

\subsection{Properties and Examples}
\begin{lm}
The robust bubble $\beta$ is a positive $\QG$-local submartingale for every $\QG\in\QC$, such that $\beta_T=0$ q.s. Moreover, if there exists a bubble, $S$ is not a $\QC$-martingale.
\end{lm}

\begin{proof}
This immediately follows from Definition \ref{rfv}, as $\beta$ is the difference between a $\QG$-local martingale and a $\QG$-supermartingale. 
\end{proof}

\noindent
The local submartingale characterization is not at all a contradiction as it could seem at a first sight. In fact, as opposed to local supermartingales, local submartingales do not have to be true submartingales. There are a variety of examples of local submartingales with nonstandard behavior, such as decreasing mean, as shown in \cite{yor_strictlocal} and \cite{protter_htransform}. To mention a clear example, it suffices to consider the class of positive local martingales: such processes are positive local submartingale and also supermartingales. \\
We can present one first example of robust bubble by adapting one result of \cite{hobson_bubbles} to the context of $G$-expectation.

\begin{rem}
We highlight that in Example \ref{ex1}, Example \ref{ex2} and Example \ref{ex3} the asset price $S$ is a $\QG$-local martingale for every $\QG\in\QC$ under the completed filtration $\FG^\QG=\{\FC_t^\QG\}_{t\geq 0}$. However, being a positive process adapted to $\FG^\ast\subseteq\FG^\QG$, $S$ is also a $\QG$-local martingale with respect to the filtration $\FG^\ast$, thanks to a result from \cite{stricker} that we report in the formulation of Theorem 10 from \cite{protter_shrink}.
\begin{thm}
Let $X$ be a positive local martingale for $\GG$ and assume that $X$ is adapted to the subfiltration $\FG$. Then $X$ is also a local martingale for $\FG$.
\end{thm}
\end{rem}

\begin{ex}\label{ex1}
Let $\QC= \PC_{\textbf{D}}$ as in Proposition \ref{G}, where $\textbf{D}=[\underline{\sigma}^2,\overline{\sigma}^2]\subset \mathbb{R}_+\setminus \{0\}$. Let $S_0=s>0$ and 
\beq
S_t=s+\int_0^t \frac{S_u}{\sqrt{T-u}}dB_u,\quad t\in[0,T).
\eeq
We show that $S$ is a price process with a robust bubble by showing that $S$ is a positive $\QG$-local martingale for every $\QG\in\QC$, with terminal value equal to zero. To this purpose, let us fix a prior $\QG$.
We have that 
\[
S_t=se^{\int_0^t\varphi_sdB_s-\frac{1}{2}\int_0^t\varphi_sd\langle B\rangle_s},\quad t\in[0,T).
\]
The stochastic integral $\int_0^\cdot 1/ \sqrt{T-s} dB_s$ is a $\QG$-local martingale on $[0,T)$, such that the quadratic covariation is
\[
\left[ \int_0^\cdot 1/ \sqrt{T-s} dB_s,\int_0^\cdot 1/ \sqrt{T-s} dB_s\right]_u\geq -\underline{\sigma}^2\ln\left[1-\frac{u}{T}\right],
\]
and continuous on $[0,T)$. Using the same argument as in Lemma 5 from \cite{protter_complete}, which exploits the Dubins-Schwarz theorem together with the law of the iterated logarithm, we can argue that 
\beq\label{asconv}
\lim_{u\to T} S_u=0\qquad \QG-a.s.
\eeq
Hence we set $S_T=0$ so that $S$ is quasi continuous on $[0,T]$. This follows as the set $\{\omega\in\Omega:\;\lim_{u\to T}S_u\neq 0\}$ is polar: the existence of $\QG\in\QC$ such that $\QG(\lim_{u\to T}S_u\neq 0)>0$ is in fact in contradiction with \eqref{asconv}. Hence $S$ is not a $\QG$-martingale for any $\QG\in\QC$ as $E_\QG[S_T]=0<E_\QG[S_0]$, and in particular it is not a robust martingale.
\end{ex}

\noindent
Another example comes from adapting the concept of Bessel process in the context of Proposition \ref{const}.

\begin{ex}\label{ex2}
We consider $\QC=\PC_{\textbf{D},const}$, where $\textbf{D}\subset\RG^{3\times 3}$ is made of those matrices $(a_{i,j})_{i,j=1,2,3}$ such that $a_{i,j}=0$ for all $i\neq j$, $a_{1,1}=a_{2,2}=a_{3,3}\in [1,2]$ in order to fix some values.\\
We consider the process given by $f(B)=(f(B_t))_{t\geq 0}$ where $f(x,y,z)=(x^2+y^2+z^2)^{-\frac{1}{2}}$. As $f$ is Borel-measurable, we can compute the sublinear expectation $\EC_0(f(B_t))$ for any $t\geq 0$, according to Theorem \ref{robustsetting}. It is a well known result that $f(B)$ is a $\QG$-local martingale for all $\QG\in\PC_{\textbf{D},const}$. To prove that the price process has a robust bubble, it thus suffices to show that $f(B)$ is not a $\PC_{\textbf{D},const}$-martingale. This can be done using a standard argument, for which we first compute
\begin{align*}
\mathcal{E}_0\left(f^2(B_t)\right)=&\sup_{\QG\in\PC_{\textbf{D},const}}E_\QG\left[f^2(B_t)\right]\\
&=\sup_{a\in[1,2]}\frac{1}{(2\pi at)^{3/2}}\int_{\RG^3}\frac{1}{x_1^2+x_2^2+x_3^2}\,\exp\left(-\frac{x_1^2+x_2^2+x_3^2}{2at}\right)\,d x_1d x_2d x_3\\
&\leq C \frac{1}{t},
\end{align*}
for some $C\in\RG_+$. As 
\[
E_\QG\left[f^2(B_t)\right]\geq E_\QG\left[f(B_t)\right]^2,
\]
for any $\QG\in\PC_{\textbf{D},const}$ implies
\[
\sup_{\QG\in\PC_{\textbf{D},const}}E_\QG\left[f^2(B_t)\right]\geq \sup_{\QG\in\PC_{\textbf{D},const}}E_\QG\left[f(B_t)\right]^2,
\]
we have 
\[
0\leq \EC_0\left[f(B_t)\right]^2\leq  \EC_0\left[f^2(B_t)\right]\to0,
\]
as $t\to\infty$. This prevents $\EC_0\left(f(B_t)\right)$ to be constant and thus $f(B)$ from being a $\PC_{\textbf{D},const}$-martingale.
\end{ex}

\noindent
In both Example \ref{ex1} and Example \ref{ex2}, the risky asset is a strict $\QG$-local martingale for every $\QG\in\QC$. This means that the bubble in those cases is perceived under all priors which are possible in the model.\\
By slightly modifying the framework of Example \ref{ex2}, we are able to obtain a bubble that is a $\bar{\QG}$-martingale for a particular $\bar{\QG}\in\QC$. Thus, despite the bubble being robust, an investor endowed only with the prior $\bar{\QG}$ would not detect it. This is one of the main novelty of our model: when a bubble arises it will be identified by an agent whose \emph{significative sets} are those with positive probability under any $\QG\in\QC$; alternatively stated this agent considers negligible only the \emph{polar sets}, i.e.\ those $A\in\FC$ such that $\QG(A)=0$ for all $\QG\in\QC$. However a \emph{short-sighted} investor, who neglects only the $\bar{\QG}$-null sets, will not spot the bubble. 

\begin{ex}\label{ex3}
We consider $\QC=\PC_{\textbf{D},const}$ as in Example \ref{ex2} but now we choose $\textbf{D}$ in a way to allow for a degenerate case, where there exists $\bar{Q}\in\QC$ such that the canonical process is constantly equal to $0$. We do this by considering the same setup as in Example \ref{ex2}, but choosing $a_{1,1}=a_{2,2}=a_{3,3}\in [0,2]$. Exactly as in Example \ref{ex2}, $f(B)$ is a $\QG$-local martingale for every $\QG\in\QC$. However under the `degenerate prior' $\bar{\QG}$, associated to a volatility constantly equal to $0$, every process turns deterministic. This implies in particular that $f(B)$ is a true $\bar{\QG}$-martingale, while being a strict $\QG$-local martingale for all $\QG\in\QC\setminus\{\bar{\QG}\}$.
\end{ex}

\noindent
The examples regarding financial bubbles are usually obtained by showing specific asset dynamics with strict local martingale behavior. We give here an example of robust bubble by rather focusing our attention on the choice of probability measures included in the uncertainty framework. 

\begin{ex}\label{ex4}
We adopt here the financial model introduced in \cite{nutz_superhedging}. The major difference with respect to the setting presented in Section \ref{setting} is that we consider the set $\QC_S$ of laws 
\begin{equation}\label{ns}
\QG^\alpha:=\QG_0\circ (X^\alpha)^{-1},\quad\text{ where } \quad X_t^\alpha:=\int_0^t \alpha_s^{1/2}dB_s,\quad t\in[0,T].
\end{equation}
In \eqref{ns} $\QG_0$ denotes the Wiener measure, while $\alpha$ ranges over all the $\FG$-progressively measurable processes with values in $\mathbb{S}_d^{+}$ satisfying $\int_0^T |\alpha_s| ds<\infty$ $\QG_0$-a.s. Here $\mathbb{S}_d^{+}\subset \mathbb{R}^{d\times d}$ represents the set of all strictly positive definite matrices and the stochastic integral in \eqref{ns} is the It\^o integral under $\QG_0$. The set $\QC$ is asked to be stable under pasting, according to the following definition.

\begin{df}\label{paste}
The set $\QC$ is stable under $\FG$-pasting if for all $\QG\in\QC$, $\sigma$ stopping time taking finitely many values, $\Lambda\in \FC_\sigma$ and $\QG_1,\QG_2\in\QC(\sigma, \QG)$, the measure $\bar{\QG}$ defined by 
\begin{equation}\label{pastingm}
\bar{\QG}(A):= E_\QG \left[ \QG_1(A|\FC_\sigma){\bf{1}}_{\Lambda}+\QG_2(A|\FC_\sigma){\bf{1}}_{\Lambda^c}\right],\quad A\in\FC_T
\end{equation}
is again an element of $\QC$.
\end{df}

\noindent
Besides from that we leave all the other definitions stated in the preceding sections unchanged. Let then be given a risky asset $S$ such that there exists a $\QG^{\tilde{\alpha}}\in\QC_S$ for which the asset is a strict $\QG^{\tilde{\alpha}}$-local martingale. For example we can take $S$ to be the process described in Example \ref{ex1}. We then study how this fact can generate a robust bubble. To do this we consider a subset $\QC\subseteq\QC_S$ given by those $\QG^\alpha\in\QC_S$ for which 
\[
\alpha_s=\tilde{\alpha}_s \quad  \text{for }\; s\in(t,T]\;\; \QG_0-a.s. 
\]
for some $t\in(0,T)$. With such requirement, the set $\QC$ is stable under pasting, according to Definition \ref{paste}, thanks to the same proof of Lemma 3.3 in \cite{nutz_superhedging}. In other words, we are considering a subset of $\QC_S$ where there is no uncertainty after time $t$, and where the volatility on $(t,T]$ implies a strict local martingale behavior under at least one prior. It follows that, for every $s>t$,
\[
\esssup_{\QG\in\QC(s,\QG^{\tilde{\alpha}})} E_\QG [S_T|\FC_s] = E_{\QG^{\tilde{\alpha}}} [S_T|\FC_s] < S_s,
\]
thus implying the presence of a robust bubble.
\end{ex}

\noindent
We now investigate another interesting relation between robust and classical bubbles. The arguments in Section \ref{RFV} clarified that the existence of $\QG\in\QC$ and $t\in[0,T)$ such that  
\[
S_t> \esssup_{\QG^\prime\in\QC(t,\QG)} E_{\QG^\prime}[S_T|\FC_t]
\]
implies the presence of a bubble in the classical sense for all the $\QG^\prime$-markets with $\QG^\prime\in\QC(t,\QG)$. This is evident at least for two situations: when every $\QG$-market admits a unique ELMM or when fundamental prices are described as the expected value of future discounted payoffs. In these two cases we prove that a single classical bubble cannot generate a robust bubble, which corresponds to intuition. We do that by showing that set of priors $\QG\in\QC$ for which $S$ is a strict local martingale cannot be a singleton. To show this result we consider the setting outlined in Example \ref{ex4}. This choice allows at the same time to ease the computations and to infer some conclusions about the framework outlined in Section \ref{setting}, as both models can describe the $G$-setting.

\begin{prop}\label{pastingprop}
Consider the financial model introduced in Example \ref{ex4}. If $\bar{\QG}$ is the pasting of $\QG$, $\QG_1$ and $\QG_2$ at the stopping time $\sigma$ and $\Lambda\in\FC_\sigma$, as in \eqref{pastingm}, it holds
\begin{equation}
\label{claim1}
E_{\bar{\QG}}[Y|\FC_\tau]=E_\QG\left[ E_{\QG_1}[Y{\bf{1}}_\Lambda|\FC_\sigma]|\FC_\tau\right]+E_\QG\left[E_{\QG_2}[Y{\bf{1}}_{\Lambda^c}|\FC_\sigma]|\FC_\tau\right]
\end{equation}
for any positive $\FC_T$-measurable random variable $Y$ and stopping time $\tau$ such that $\tau(\omega)\leq \sigma(\omega)$ for every $\omega\in\Omega$.
\end{prop}

\begin{proof}
We follow a procedure similar to Lemma 6.40 in \cite{follmer_book} to prove \eqref{claim1}. Let $\tau$ be a stopping time and $Y$ a positive $\FC_T$-measurable random variable. By \eqref{pastingm} we have that 
\[
E_{\bar{\QG}}[Y]=E_\QG\left[ E_{\QG_1}[Y|\FC_\sigma]{\bf{1}}_\Lambda+E_{\QG_2}[Y|\FC_\sigma]{\bf{1}}_{\Lambda^c}\right],
\]
so that, for every positive $\FC_\tau$-measurable random variable $\varphi$, we can study the value of
\begin{equation}\label{past1}
E_{\bar{\QG}}[Y\varphi\ind{\tau\leq \sigma}].
\end{equation}
The expectation in \eqref{past1} can then be written as
\begin{align}
\notag E_{\bar{\QG}}[Y\varphi\ind{\tau\leq \sigma}]&=E_\QG\left[ E_{\QG_1}[Y\varphi\ind{\tau\leq \sigma}|\FC_\sigma]{\bf{1}}_\Lambda+E_{\QG_2}[Y\varphi\ind{\tau\leq \sigma}|\FC_\sigma]{\bf{1}}_{\Lambda^c}\right]\\
\notag &=E_\QG\left[ E_{\QG_1}[Y\varphi\ind{\tau\leq \sigma}{\bf{1}}_\Lambda|\FC_\sigma]+E_{\QG_2}[Y\varphi\ind{\tau\leq \sigma}{\bf{1}}_{\Lambda^c}|\FC_\sigma]\right]\\
\label{past2} &=E_\QG\Big[E_\QG\left[ E_{\QG_1}[Y{\bf{1}}_\Lambda|\FC_\sigma]|\FC_\tau\right]\varphi\ind{\tau\leq \sigma}+E_\QG\left[E_{\QG_2}[Y{\bf{1}}_{\Lambda^c}|\FC_\sigma]|\FC_\tau\right]\varphi\ind{\tau\leq \sigma}\Big]\\
\notag &=E_{\bar{\QG}}\Big[E_\QG\left[ E_{\QG_1}[Y{\bf{1}}_\Lambda|\FC_\sigma]|\FC_\tau\right]\varphi\ind{\tau\leq \sigma}+E_\QG\left[E_{\QG_2}[Y{\bf{1}}_{\Lambda^c}|\FC_\sigma]|\FC_\tau\right]\varphi\ind{\tau\leq \sigma}\Big]\\
\notag &=E_{\bar{\QG}}\left[\left(E_\QG\left[ E_{\QG_1}[Y{\bf{1}}_\Lambda|\FC_\sigma]|\FC_\tau\right]+E_\QG\left[E_{\QG_2}[Y{\bf{1}}_{\Lambda^c}|\FC_\sigma]|\FC_\tau\right]\right)\varphi\ind{\tau\leq \sigma}\right].
\end{align}

\noindent
Hence we can conclude that if $\tau\leq \sigma$ the equality \eqref{claim1} holds. 
\end{proof}

\begin{cor}
Consider $\bar{\QG}$ given by the pasting of $\QG$, $\QG_1$ and $\QG_2$ at the stopping time $\sigma$ and $\Lambda\in\FC_\sigma$, as in \eqref{pastingm}. If $S$ is a strict $\QG_1$-local martingale, then it is also a strict $\bar{\QG}$-local martingale.
\end{cor}

\begin{proof}
As a consequence of Proposition \ref{pastingprop}, if $S$ is a strict $\QG_1$-local martingale and $\sigma$ is such that 
\[
E_{\QG_1}[S_T|\FC_\sigma]<S_\sigma,
\]
it holds
\begin{align*}
E_{\bar{\QG}}[S_T|\FC_\tau]&=E_\QG\left[ E_{\QG_1}[S_T{\bf{1}}_\Lambda|\FC_\sigma]|\FC_\tau\right]+E_\QG\left[E_{\QG_2}[S_T{\bf{1}}_{\Lambda^c}|\FC_\sigma]|\FC_\tau\right]\\
&=E_\QG\left[ E_{\QG_1}[S_T|\FC_\sigma]{\bf{1}}_\Lambda|\FC_\tau\right]+E_\QG\left[E_{\QG_2}[S_T|\FC_\sigma]{\bf{1}}_{\Lambda^c}|\FC_\tau\right]\\
&<E_\QG\left[ S_\sigma{\bf{1}}_\Lambda|\FC_\tau\right]+E_\QG\left[S_\sigma{\bf{1}}_{\Lambda^c}|\FC_\tau\right]\\
& \leq  S_\tau,
\end{align*}
so that $S$ is also a strict $\bar{\QG}$-local martingale.
\end{proof}

\subsection{No Dominance}
In this section we investigate the implications of no dominance in our market model. This is a concept first appeared in \cite{merton}, that we report in the rigorous mathematical form stated in Definition 2.2.\ of \cite{marketefficiency} for the classical situation in which a unique prior $\QG$ exists.

\begin{df}
Let be given a financial market with $d$ securities $(S^1,\dots,S^d)$ in a filtered probability space $(\Omega,\FC,\FG=\{\FC_t\}_{t\in[0,T]},\PG)$. $H$ is an admissible strategy if it is an $\FG$-predictable and $S$-integrable process such that $H\cdot S \geq -a$, for some $a\in\mathbb{R}_+$. We say that the $i$-th security $S^i$ is undominated on $[0,T]$ if there is no admissible strategy $H$ such that 
\[
S^i_0+(H	\cdot S)_T\geq S^i_T\;\; \QG-a.s.\quad\text{ and }\quad \QG(S^i_0+(H	\cdot S)_T> S^i_T)>0.
\]
A market satisfies no dominance (ND) on $[0,T]$ if each $S^i$, $i\in\{1,\dots,d\}$, is undominated on $[0,T]$.
\end{df}

\noindent
It is natural to transpose this concept to our setting with uncertainty, as we do in the following definition.

\begin{df}
Consider a market model under a set of priors $\QC$. The $i$-th security $S^i$ is \emph{undominated} on $[0,T]$ if there is no admissible strategy $H\in\HC$ such that 
\[
S^i_0+(H	\cdot S)_T\geq S^i_T\;\; \QC-q.s.\text{ and there exists a $\QG\in\QC$ such that } \QG(S^i_0+(H	\cdot S)_T> S^i_T)>0.
\]
A market satisfies \emph{robust no dominance} (RND) on $[0,T]$ if each $S^i$, $i\in\{1,\dots,d\}$, is undominated on $[0,T]$.
\end{df}

\begin{rem}
It is important to notice that, as in the classical case, if $S^i$ is undominated on $[0,T]$, it also undominated on $[0,T^\prime]$, for $T^\prime<T$. Let $H^i$ be given by 
\[
H^i=(0,\dots,0,1,0,\dots,0),
\]
with $1$ in position $i$. The trading strategy $H^i$ is admissible, being $S$ a $\QG$-local martingale for every $\QG\in\QC$. If there would be a dominating strategy $H$ on $[0,T^\prime]$, by applying the strategy $K=H\ind{t\leq T^\prime} + H^i\ind{t>T^\prime}$, we would obtain
\[
S^i_0+(K\cdot S)_T=S_T^i+S^i_0+(H\cdot S)_{T^\prime}-S^i_{T^\prime} \geq S_T^i\quad q.s.,
\]
together with the existence of a $\QG\in\QC$ such that 
\[
\QG(S^i_0+(K\cdot S)_T>S_T^i)>0.
\]
\end{rem}

\noindent
The ND assumption plays a key role in the classical literature on bubbles. We just mention two results by reminding that, if enforced, this concept rules out bubbles in the complete market models described by \cite{protter_complete}; moreover, ND is precisely the ingredient needed to exclude bubbles in the setting of \cite{schweizer_bubbles}, where fundamental values are modeled with superreplication prices. Similar results can be obtained also in the present framework.
\begin{lm}
Suppose that for each $\QG\in\QC$ the $\QG$-market model is complete. If robust no dominance holds, then there exists no robust bubble.
\end{lm}

\begin{proof}
Observe that, if each $\QG$-market is complete, the results of \cite{nutz_jumps} guarantee the duality
\beq
\label{dual}
\sup_{\QG\in\QC} E_\QG[S_T]= \inf\{x\in\RG\;:\;\exists \;H\in\mathcal{H} \text{ with } x+(H\cdot S)_T\geq S_T\; \QG-a.s. \text{ for all } \QG\in\QC\}.
\eeq
In presence of a bubble, the superreplicating strategy would then dominate $S$, in contradiction with RND.
\end{proof}

\noindent
Hence, in the general case, under RND any bubble would be the result of a duality gap in \eqref{dual}, which is the case considered in \cite{cox_robust}. 

We remark how in general RND does not imply NFLVR for every $\QG$-market, $\QG\in\QC$. It is in fact well known that ND is stronger than NFLVR in the single prior setting, but it is a priori not necessary that RND implies ND for every $\QG$-market.

\section{Infinite Time Horizon}
We study here the case of infinite time horizon. Let $\tau>0$ q.s.\ be a stopping time describing the maturity of the risky asset. To reflect the impossibility of the investor to consume the final payoff of $S$ in the case $\{\tau=\infty\}$ we generalize here the robust fundamental value established in \eqref{rfve} by setting
\beq\label{rfv2e}
S^{\ast}_t=\left(\esssup_{\mathbb{Q}^\prime\in \mathcal{Q}(t,\mathbb{Q})} E_{\mathbb{Q}^\prime}[S_\tau\ind{\tau<\infty}|\mathcal{F}_t]\right)\ind{t<\tau},\quad \QG-a.s.
\eeq
for every $t\geq 0$ and $\QG\in\QC$. The fundamental value \eqref{rfv2e} embodies the finite time horizon case \eqref{rfve} and we claim that it is well defined, as shown in the following proposition.

\begin{prop}\label{infinity}
The fundamental value \eqref{rfv2e} is well defined. In addition, $S_{t\wedge\tau}$ converges to $S_\tau$ q.s.\  for $t\rightarrow\infty$.
\end{prop}

\begin{proof}
Fixed $\QG\in\QC$, we know that $W$ is a $\QG$-supermartingale, which then converges $\QG$-a.s.\ to $S_\tau$ for $t\to\infty$, because of the classical supermartingale convergence theorem (see \cite{dellacherie}, V.28 and VI.6). Therefore $W_t=S_{t\wedge \tau}\rightarrow S_\tau$ q.s., thanks to the same argument used in Example \ref{ex1}, and $S_\tau$ is Borel measurable. Hence $S_\tau\ind{\tau<\infty}$ is a Borel measurable random variable and we can compute its sublinear conditional expectation. Moreover, as $W$ is a robust supermartingale, by Fatou's Lemma we obtain
\[
\begin{split}
\EC_0(S_{\tau})&=\EC_0\left(\liminf_{t\to\infty} S_{t\wedge \tau}\right)=\sup_{\QG\in\QC} E_\QG\left(\liminf_{t\to\infty}S_{t\wedge \tau}\right)\leq  \sup_{\QG\in\QC} \liminf_{t\to\infty}E_\QG\left(S_{t\wedge \tau}\right)\\
&=\sup_{\QG\in\QC} \liminf_{t\to\infty}E_\QG\left(W_{t}\right)\leq \sup_{\QG\in\QC} E_\QG\left(W_{0}\right)<\infty,
\end{split}
\]
which guarantees $\EC_0(S_\tau\ind{\tau<\infty})<\infty$. 
\end{proof}

\noindent 
We introduce also the notion of robust fundamental wealth, by defining the process $W^\ast=(W^\ast_t)_{t\geq 0}$, where 
\begin{align}
\notag W_t^\ast:&=S_t^\ast+S_\tau\ind{\tau\leq t}=\left(\esssup_{\mathbb{Q}^\prime\in \mathcal{Q}(t,\mathbb{Q})} E_{\mathbb{Q}^\prime}[S_\tau\ind{\tau<\infty}|\mathcal{F}_t]\right)\ind{t<\tau}+S_\tau\ind{\tau\leq t}\\
\label{infinite} &=\esssup_{\mathbb{Q}^\prime\in \mathcal{Q}(t,\mathbb{Q})} E_{\mathbb{Q}^\prime}[S_\tau\ind{\tau<\infty}|\mathcal{F}_t],\qquad \QG-a.s.
\end{align}
for all $\QG\in\QC$. A robust bubble is defined as in the finite time horizon case, i.e.\
\[
\beta_t= S_t-S^\ast_t=W_t-W^\ast_t,
\]
for every $t\geq 0$. As a consequence, the case $\tau=\infty$ q.s.\ implies the presence of a robust bubble. As argued in \cite{protter_complete}, the bubble appearing in this situation is analogous to fiat money, a terminal value obtained at $\infty$. We report here Example 2 from \cite{protter_complete} to clarify this point.

\begin{ex}\label{fiat}
Let $S_t=1$ for all $t\in\mathbb{R}_+$ be fiat money. Since money never matures, we have $\tau=\infty$, $S_\tau=1$ and $S_t^\ast=0$ q.s.\ for all $t\geq 0$. As
\[
\beta_t=S_t-S_t^\ast=1 \quad q.s.
\]
this means that the entire value of the asset comes from the bubble.
\end{ex}

\noindent
We summarize these results in the following proposition.

\begin{prop}
It holds:
\begin{itemize}
\item[(i)] In the case there exists a $\bar{\QG}\in\QC$ and $t\geq 0$ such that $\bar{\QG}^\prime(\tau=\infty)=1$ for all $\bar{\QG}^\prime\in\QC(t,\bar{\QG})$, there exists a robust bubble.
\item[(ii)] The bubble $\beta$ is a $\QG$-local submartingale for every $\QG\in\QC$.
\item[(iii)] The wealth process $W$ can be a $\QC$-symmetric martingale also in the presence of a bubble. 
\end{itemize}
\end{prop}

\begin{proof}
The proof of (i) follows from \eqref{infinite}, noticing that 
\[
W_t^\ast=0 \qquad \bar{\QG}-a.s.,
\]
as by hypothesis
\[
S_\tau\ind{\tau<\infty}=0 \qquad \bar{\QG}^\prime-a.s. 
\]
for all $\bar{\QG}^\prime\in\QC(t,\bar{\QG})$. The local submartingale property follows from the definition and from Assumption \ref{NFLVR}. The wealth process can be a $\QC$-symmetric martingale as it can be seen in Example \ref{fiat}.
\end{proof}

\end{document}